\documentclass[11pt]{article}
\usepackage{url,ifthen}
\usepackage{srcltx}
\usepackage{multirow}
\usepackage{boxedminipage}
\usepackage[margin=1.1in]{geometry}
\usepackage{nicefrac}
\usepackage{xspace}
\usepackage{graphicx}
\usepackage{srcltx}
\usepackage[usenames]{color}
\usepackage{fullpage}
\usepackage{xspace}
\definecolor{DarkGreen}{rgb}{0.1,0.5,0.1}
\definecolor{DarkRed}{rgb}{0.5,0.1,0.1}
\definecolor{DarkBlue}{rgb}{0.1,0.1,0.5}
\usepackage[small]{caption}
\usepackage{float}
\usepackage{balance}
\usepackage{amsmath}
\usepackage{amsfonts}
\usepackage{amssymb}
\usepackage{amsthm}

\newtheorem{theorem}{Theorem}[section]
\newtheorem*{namedtheorem}{\theoremname}
\newcommand{\theoremname}{testing}

\newtheorem{lemma}[theorem]{Lemma}

\newtheorem{claim}[theorem]{Claim}

\newtheorem{observation}{Observation}
\newtheorem{corollary}[theorem]{Corollary}

\newtheorem{question}[theorem]{Question}

\newtheorem*{question*}{Question}

\theoremstyle{definition}
\newtheorem{definition}[theorem]{Definition}

\theoremstyle{plain}
\newtheorem{Alg}{Algorithm}

\definecolor{DarkGreen}{rgb}{0.1,0.5,0.1}
\definecolor{DarkRed}{rgb}{0.5,0.1,0.1}
\definecolor{DarkBlue}{rgb}{0.1,0.1,0.5}

\usepackage[pdftex]{hyperref}
\hypersetup{
    unicode=false,          
    pdftoolbar=true,        
    pdfmenubar=true,        
    pdffitwindow=false,      
    pdfnewwindow=true,      
    colorlinks=true,       
    linkcolor=DarkGreen,          
    citecolor=DarkGreen,        
    filecolor=DarkGreen,      
    urlcolor=DarkBlue,          
}

\newcommand{\ignore}[1]{}

\newcommand{\E}{\mathop{\bf E\/}}

\makeatletter
\floatstyle{ruled}
\newfloat{fragment}{H}{lop}
\floatname{fragment}{Algorithm}
\renewcommand{\floatc@ruled}[2]{\vspace{2pt}{\@fs@cfont \#1.\:} \#2 \par
 \vspace{1pt}}
\makeatother

\title{Approximate Counting, the Lov\'asz Local Lemma and Inference in Graphical Models}
\author{Ankur Moitra\thanks{
Massachusetts Institute of Technology. Department of Mathematics and the Computer Science and Artificial Intelligence Lab. Email: {\tt moitra@mit.edu}. This work was supported in part by NSF CAREER Award CCF-1453261, NSF Large CCF-1565235, a David and Lucile Packard Fellowship, an Alfred P. Sloan Fellowship, an Edmund F. Kelley Research Award, a Google Research Award and a grant from the MIT NEC Corporation.} }

\begin{document}
\maketitle

\begin{abstract}
\normalsize
In this paper we introduce a new approach for approximately counting in bounded degree systems with higher-order constraints. Our main result is an algorithm to approximately count the number of solutions to a CNF formula $\Phi$ when the width is logarithmic in the maximum degree. This closes an exponential gap between the known upper and lower bounds. 

Moreover our algorithm extends straightforwardly to approximate sampling, which shows that under Lov\'asz Local Lemma-like conditions it is not only possible to find a satisfying assignment, it is also possible to generate one approximately uniformly at random from the set of all satisfying assignments. Our approach is a significant departure from earlier techniques in approximate counting, and is based on a framework to bootstrap an oracle for computing marginal probabilities on individual variables. Finally, we give an application of our results to show that it is algorithmically possible to sample from the posterior distribution in an interesting class of graphical models. 
\end{abstract}

\thispagestyle{empty}

\newpage

\setcounter{page}{1}

\section{Introduction}

\subsection{Background}

In this paper we introduce a new approach for approximately counting in bounded degree systems with higher-order constraints. For example, if we are given a CNF formula $\Phi$ with $n$ variables and $m$ clauses with the property that each clause contains between $k$ and $2k$ variables and each variable belongs to at most $d$ clauses we ask:

\begin{question}
How does $k$ need to relate to $d$ for there to be algorithms to estimate the number of satisfying assignments to $\Phi$ within a $(1\pm1/n^c)$ multiplicative factor?
\end{question}

In the case of a monotone CNF formula where no variable appears negated, the problem is equivalent to the following: Suppose we are given a hypergraph on $n$ nodes and $m$ hyperedges with the property that each hyperedge contains between $k$ and $2k$ nodes and each node belongs to at most $d$ hyperedges. How does $k$ need to relate to $d$ in order to be able to approximately compute the number of independent sets? Here an independent set is a subset of nodes for which there is no induced hyperedge. Bordewich, Dyer and Karpinski \cite{BDK1} gave an MCMC algorithm for approximating the number of hypergraph independent sets (equivalently, the number of satisfying assignments in a monotone CNF formula) that succeeds whenever $k \geq d + 2$. Bez\'akova et al. \cite{BGGGS} gave a deterministic algorithm that succeeds whenever $k \geq d  \geq 200$ and proved that when $d \geq 5 \cdot 2^{k/2}$ it is $NP$-hard to approximate the number of hypergraph independent sets even within an exponential factor. 

More broadly, there is a rich literature on approximate counting problems. In a seminal work, Weitz \cite{W} gave an algorithm to approximately count
in the hardcore model with parameter $\lambda$
in graphs of degree at most $d$ whenever $$\lambda \leq \frac{(d-1)^{d-1}}{(d-2)^d}$$ And in another seminal work, Sly \cite{S} showed a matching hardness result which was later improved in various respects by Sly and Sun \cite{SS} and Galanis, \u{S}tefankovi\u{c} and Vigoda \cite{GSV}. These results show that approximate counting is algorithmically possible if and only if there is spatial mixing. Moreover, Weitz's result can be thought of as a comparison theorem that spatial mixing holds on a bounded degree graph if and only if it holds on an infinite tree with the same degree bound. There have been a number of attempts to generalize these results to hypergraphs, many of which follow the approach of defining analogues of the self-avoiding walk trees used in Weitz's algorithm \cite{W}. However what makes hypergraph versions of these problems more challenging is that spatial mixing fails, even on trees. And we can see that there are {\em exponential} gaps between the upper and lower bounds, since the algorithms above require $k$ to be linear in $d$ while the lower bounds only rule out $k \leq 2 \log d - O(1)$. 

We can take another vantage point to study these problems. Bounded degree CNF formulae are also one of the principal objects of study in the Lov\'asz Local Lemma \cite{EL} which is a celebrated result in combinatorics that guarantees when $k \geq \log d + \log k + O(1)$ that $\Phi$ has at least one satisfying assignment. The original proof of the Lov\'asz Local Lemma was non-constructive and did not yield a polynomial time algorithm for finding such an assignment, even though it was guaranteed to exist. Beck \cite{B} gave an algorithm followed by a parallel version due to Alon \cite{A} that can find a satisfying assignment whenever $k \geq 8 \log d + \log k + O(1)$. And in a celebrated recent result, Moser and Tardos \cite{MT} gave an algorithm exactly matching the existential result. This was followed by a number of works giving constructive proofs of various other settings and generalizations of the Lov\'asz Local Lemma \cite{HS, AI, HV, Ko}. However these works leave open the following question:

\begin{question}
Under the conditions of the Lov\'asz Local Lemma (i.e. when $k$ is logarithmic in $d$) is it possible to approximately sample from the uniform distribution on satisfying assignments? 
\end{question}

Approximate counting and approximate sampling problems are well-known to be related. When the problem is self-reducible, they are in fact algorithmically equivalent \cite{JVV, JS}. However in our setting the problem is not self-reducible because as we fix variables we could violate the assumption that $k$ is at least logarithmic in $d$. It is natural to hope that under exactly the same conditions as the Lov\'asz Local Lemma, that there is an algorithm for approximate sampling that matches the limits of the existential and now algorithmic results. However the hardness results of Bez\'akova et al. \cite{BGGGS} imply that we need at least another factor of two, and that it is $NP$-hard to approximately sample when $k \leq 2 \log d - O(1)$\footnote{The hardness results in \cite{BGGGS} are formulated for approximate counting but carry over to approximate sampling. In particular, an oracle for approximately sampling from the set of satisfying assignments yields an oracle for approximating the marginal at any variable. Then one can invoke Lemma $7$ in \cite{BGGGS}.}. 

In fact, there is another connection between the Lov\'asz Local Lemma and approximate counting. Scott and Sokal \cite{SS} showed that given the dependency graph of events in the local lemma, the best lower bound on the probability of an event guaranteed to exist by the Lov\'asz Local Lemma (i.e. the fraction of satisfying assignments) is exactly the solution to some counting problem. Harvey, Srivastava and Vondr\'ak \cite{HSV} recently adapted techniques of Weitz to complex polydisks and gave an algorithm for approximately computing this lower bound. This yields a lower bound on the fraction of satisfying assignments, however the actual number could be exponentially larger. 

\subsection{Our Results}

Our main result is an algorithm to approximately count the number of solutions when $k$ is at least logarithmic in $d$. In what follows, let $c$, $k$ and $d$ be constants. We prove\footnote{We have not made an attempt to optimize the constant in this theorem. }:

\begin{theorem}[informal]
Suppose $\Phi$ is a CNF formula where each clause contains between $k$ and $2k$ variables and at most $d$ clauses containing any one variable. For $k \gtrsim 60 \log d$ there is a deterministic polynomial time algorithm for approximating the number of satisfying assignments to $\Phi$ within a multiplicative $(1\pm 1/n^c)$ factor. Moreover there is a randomized polynomial time algorithm to sample from a distribution that is $1/n^c$-close in total variation distance to the uniform distribution on satisfying assignments. 
\end{theorem}

This algorithm closes an exponential gap between the known upper bounds \cite{BDK1, BGGGS} and lower bounds \cite{BGGGS}. It also shows that under Lov\'asz Local Lemma-like conditions not only is it possible to efficiently find a satisfying assignment, it is possible to find a random one. Moreover our approach is a significant departure from earlier techniques based either on path coupling \cite{BDK1} or adapting Weitz's approach to non-binary models and hypergraphs \cite{GK, NT, SST, LL, BGGGS}. The results above appear in Theorem~\ref{thm:nonmonotonecounting} and Theorem~\ref{thm:approximatesampling}. Moreover our techniques seem to extend to many non-binary counting problems as we explain in Section~\ref{sec:conclusion}. 

Our approach starts from a thought experiment about what we could do if we had access to a very powerful oracle that could answer questions about the marginal distributions of individual variables under the uniform distribution on satisfying assignments. We use this oracle and properties of the Lov\'asz Local Lemma (namely, bounds it gives on the marginal distribution of individual variables) to construct a coupling between two random satisfying assignments so that both agree outside some logarithmic sized component. If we knew the distribution on what logarithmic sized component this coupling procedure produces, we could brute force and find the ratio of the number of satisfying assignments with $x = T$ to the number with $x = F$ to compute marginals at $x$. However the distribution of what component the coupling produces intimately depends on the powerful oracle we have assumed that we have access to. 

Instead, we abstract the coupling procedure as a random root-to-leaf path in a tree that represents the state of the coupling. We show that at the leaves of this tree, there is a way to fractionally charge assignments where $x = T$ against assignments where $x = F$. Crucially, doing so requires only brute-force search on a logarithmic sized component. Finally, we show that there is a polynomial sized linear program to find a flow through the tree that produces an approximately valid way to fractionally charge assignments with $x = T$ against ones with $x = F$, and that any such solution {\em certifies} the correct marginal distribution. From these steps, we have thus bootstrapped an oracle for answering queries about the marginal distribution. Our main results then follow from utilizing this oracle. In settings where the problem is self-reducible \cite{JS} it is well-known how to go from knowing the marginal to approximate counting and sampling. In our setting, the problem is not self-reducible because setting variables could result in clauses becoming too small in which case $k$ would not be large enough as a function of $d$. We are able to get around this by using the Lov\'asz Local Lemma once more to find a safe ordering in which to set the variables. 

In an exciting development and just after the initial posting of this paper, Hermon, Sly and Zhang \cite{HSZ} gave an algorithm for approximately counting the number of independent sets in a $k$-uniform hypergraph of maximum degree $d$ provided that $k \geq 2 \log d + O(1)$. The techniques are entirely different and their algorithm matches the hardness result in \cite{BGGGS} up to an additive constant! It remains an interesting question to find similarly sharp phase transitions for the approximate counting problems studied here, namely for CNFs that are not necessarily monotone. And in yet another interesting direction, Guo, Jerrum and Liu \cite{GJL},  gave an algorithm based on connections to ``cycle popping" that can uniformly sample from $\Phi$ under weaker conditions on the degree but by imposing conditions on intersection properties of bad events. 

\subsection{Further Applications}

Our algorithms have interesting applications in graphical models. Directed graphical models are a rich language for describing distributions by the conditional relationships of their variables. However very little is known algorithmically about learning their parameters or performing basic tasks such as inference \cite{DL1, DL2}. In most settings, these problems are computationally hard. However we can study an interesting class of directed graphical models which we call {\em cause networks}. See Figure~\ref{fig:causenetwork}. 

\begin{figure}
\begin{center}
  \begin{minipage}[c]{6cm}
    \includegraphics[width=6cm]{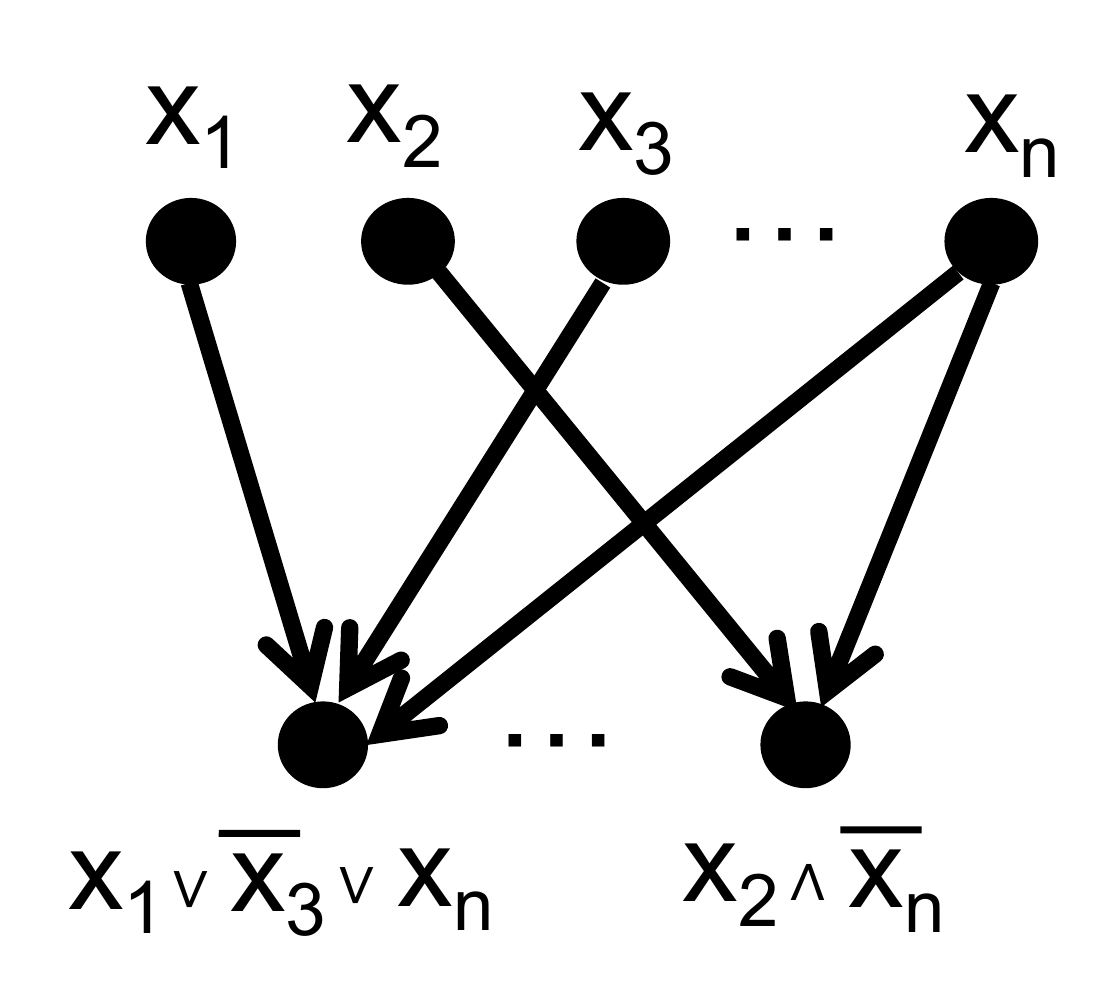}
  \end{minipage}
  \hspace{0.75cm}
  \begin{minipage}[c]{6cm}
  \vspace{1.5pc}
    \caption{
       An example cause network with $n$ hidden variables. A sample is generated by choosing each hidden variable to be $T/F$ independently with equal probability, and observing the truth values of each clause. 
    } \label{fig:causenetwork}
  \end{minipage}
  \end{center}
  \vspace{-2.0pc}
\end{figure}

\begin{definition}
In a cause network there is a collection of hidden variables $x_1, x_2, \dots, x_n$ that are chosen independently to be $T$ or $F$ with equal probability. There is a collection of $m$ observed variables each of which is either an OR or an AND of several variables or their negations. 
\end{definition}

Our goal is: Given a random sample $x_1, x_2, \dots, x_n$ from the model where we observe the truth value of each of the $m$ clauses, to sample from the posterior distribution on the hidden variables. This generalizes graphical models such as the symptom-disease network where the hidden variables represent diseases that a patient may have, and the clauses represent observed symptoms. We will require the following regularity condition on our observations: 

\begin{definition}
A collection of observations is regular if for every observed variable, the corresponding clause is adjacent to (i.e. shares a variable with) at most $15k/16$ OR clauses that are false and at most $15k/16$ AND clauses that are true. 
\end{definition}

Now, as an immediate corollary we have:

\begin{corollary}\label{corr:inference}
Given a cause network where each observed variable depends on between $k$ and $2k$ hidden variables, each hidden variable affects at most $d$ observed variables and $k \gtrsim 70 \log d$, there is a polynomial time algorithm for sampling from the posterior distribution for any regular collection of observations. 
\end{corollary}

\noindent This is a rare setting where there is an algorithm to solve an inference problem in graphical models but $(i)$ the underlying graph does not have bounded treewidth and $(ii)$ correlation decay fails. We believe that our techniques may eventually be applicable to settings where the observed variables are noisy functions of the hidden variables and where the hidden variables are not distributed uniformly.

\section{Preliminaries}

In this paper, we will be interested in approximately counting the number of satisfying assignments to a CNF formula. For example, we could be given:
 $$ \Phi = (\overline{x}_1 \vee x_3 \vee x_5) \wedge (x_2 \vee x_3 \vee \overline{x}_8) \wedge \dots \wedge (x_4 \vee \overline{x}_5 \vee x_9)$$
Let's fix some parameters. We will assume that there are $n$ variables and there are $m$ clauses each of which is an OR of between $k$ and $Ck$ distinct variables. The constant $C$ will take values either $2$ or $6$ because of the way our algorithm will be built on various subroutines. 
Finally, we will require a degree bound that each variable appears in at most $d$ clauses. We will be interested in the relationships between $k$ and $d$ that allow us to approximately count the number of satisfying assignments in polynomial time. 

The celebrated Lov\'asz Local Lemma tells us conditions on $k$ and $d$ where we are guaranteed that there is at least one satisfying assignment. Let $D$ be an upper bound on the degree of the dependency graph. We can take $D = 2dk$ or $D = 6dk$ depending on whether we are in a situation where there are at most $2k$ or at most $6k$ variables per clause. 

\begin{theorem} \cite{EL}
If $e (D+1)  \leq 2^k$ then $\Phi$ has at least one satisfying assignment. 
\end{theorem}

\noindent Moser and Tardos \cite{MT} gave an algorithm to find a satisfying assignment under these same conditions. However the assignment that their randomized algorithm finds is fundamentally not uniform from the set of all satisfying assignments. Our goal is to be able to both approximately count and uniformly sample when $k$ is logarithmic in $d$. 

There are many more related results, but we will not review them all here. Instead we state a version of the asymmetric local lemma given in \cite{HSS} which gives us some control on the uniform distribution on assignments. Let $\mathcal{C}$ be the collection of clauses in $\Phi$. Let $\mbox{Pr}[\cdot]$ denote the uniform distribution all assignments \--- i.e. uniform on $\{T, F\}^n$. Finally, for a clause $b$ let $\Gamma(b)$ denote all the clauses that intersect $b$. We can abuse notation and for any event $a$ that depends on some set of the variables, let $\Gamma(a)$ denote all the clauses that contain any of the variables on which $a$ depends.

\begin{theorem}\label{thm:distributionLLL}
Suppose there is an assignment $x: \mathcal{C} \rightarrow (0, 1)$ such that for all $c \in \mathcal{C}$ we have
$$\mbox{Pr}[c \mbox{ is unsatisfied}] \leq x(c) \prod_{b \in \Gamma(c)} \Big (1 - x(b) \Big )$$
then there is at least one satisfying assignment. Moreover the uniform distribution $\mathcal{D}$ on satisfying assignments satisfies that for any event $a$
$$\mbox{Pr}_{\mathcal{D}}[a] \leq \mbox{Pr}[a] \prod_{b \in \Gamma(a)} \Big (1 - x(b) \Big )^{-1}$$
\end{theorem}
 
 Notice that this inequality is one-sided, as it ought to be. After all if we take $b$ to be some clause, and $a$ to be the event that $b$ is not satisfied then we know that $\mbox{Pr}_{\mathcal{D}}[a] = 0$ even though $\mbox{Pr}[a]$ is nonzero. However what this theorem does tell us is that the marginal distribution of $\mathcal{D}$ on any variable is close to uniform. We will establish a quantitative version of this statement in the following corollary:

\begin{corollary}\label{corr:uniform}
Suppose that $e D s  \leq 2^k$. Then for every variable $x_i$, we have
$$ \frac{1}{2} - \frac{2}{s} \leq \mbox{Pr}_{\mathcal{D}}[x_i = T] \leq \frac{1}{2} + \frac{2}{s}$$
\end{corollary}

\begin{proof}
Set $x(c) = \frac{1}{D s}$ for each clause $c$, and consider the event $a$ that $x_i = T$. Now invoking Theorem~\ref{thm:distributionLLL} we calculate:
\begin{eqnarray*}
\mbox{Pr}_{\mathcal{D}}[x_i = T] &\leq& \mbox{Pr}[x_i = T]  \prod_{b \in \Gamma(a)} \Big (1 - x(b) \Big )^{-1} \\
&\leq & \Big (\frac{1}{2} \Big ) \Big ( 1 - \frac{1}{Ds} \Big )^{-D} \leq \frac{1}{2} + \frac{2}{s}
\end{eqnarray*}
where the last inequality follows because $(1 - \frac{1}{Ds})^{-D} \leq e^{\frac{2}{s}} \leq 1 + \frac{4}{s}$. An identical calculation works for the event $x_i = F$. All that remains is to check that the condition in Theorem~\ref{thm:distributionLLL} holds, which is a standard calculation: If $c$ is a clause then 
$$\mbox{Pr}[c \mbox{ is unsatisfied}] \leq \Big ( \frac{1}{Ds} \Big ) \Big ( 1 - \frac{1}{Ds} \Big )^D$$
The left hand side is at most $2^{-k}$ because each clause has at least $k$ distinct variables, and the right hand side is at least $(\frac{1}{Ds}) (\frac{1}{e})$. Rearranging completes the proof.
\end{proof}

\noindent Notice that $k$ is still only logarithmic in $d$ but with a larger constant, and by increasing this constant  we get some useful facts about the marginals of the uniform distribution on satisfying assignments. 

\section{A Coupling Procedure}\label{sec:couple}

\subsection{Marked Variables}

Throughout this section we will assume that the number of variables per clause is between $k$ and $6k$. Now we are almost ready to define a coupling procedure. The basic strategy that we will employ is to start from either $x = T$ and $x = F$, and then sample from the corresponding marginal distribution on satisfying assignments. If we sample a variable $y$ next, then Corollary~\ref{corr:uniform} tells us that regardless of whether $x = T$ or $x = F$, each clause has at least $k -1$ variables remaining and so the marginal distribution on $y$ is still close to uniform. 

Thus we will try to couple the conditional distributions, when starting from $x = T$ or $x = F$ as well as we can, to show that the marginal distribution on variables that are all at least some distance $\Delta$ away must converge in total variation distance. There is, however, an important catch that motivates the need for a fix. Imagine that we continue in this fashion, sampling variables from the appropriate conditional distribution. We can reach a situation where a clause $c$ has all of its variables except $y$ set and yet the clause is still unsatisfied. The marginal distribution on $y$ is no longer close to uniform. Hence, reaching small clauses is problematic because then we cannot say much about the marginal distribution on the remaining variables and it would be difficult to construct a good coupling. 

Instead, our strategy is to use the Lov\'asz Local Lemma once more, but to decide on a set of variables in advance which we call {\em marked}. 

\begin{lemma}\label{lemma:marked}
Set $c_0 = e^{(\frac{1}{2})(\frac{1}{4})^2}$. Suppose that $2 e (D+1)  \leq c_0^{k}$. Then there is an assignment $$\mathcal{M}: \{x_i\}_{i=1}^n \rightarrow \{\mbox{marked}, \mbox{unmarked}\}$$
such that for every clause $c$, it has at least $\frac{k}{4}$ marked and at least $\frac{k}{4}$ unmarked variables.
\end{lemma}

\begin{proof}
We will choose each variable to be marked or unmarked with equal probability, and independently. Consider the $m$ bad events, one for each clause $c$, that $c$ does not have enough marked or enough unmarked variables. Then we have
$$\mbox{Pr}[c \mbox{ is bad}] \leq 2 e^{-(\frac{1}{2})(\frac{1}{4})^2 k} = 2 c_0^{-k}$$
which follows from the Chernoff bound. Now we can appeal to the Lov\'asz Local Lemma to get the desired conclusion. 
\end{proof}

\noindent Only the variables that are marked will be allowed to be set to either $T$ or $F$ by the coupling procedure. The above lemma guarantees that every clause $c$ always has enough remaining variables that can make it true that the marginal distribution on any marked variable always is close to uniform. 

\subsection{Factorizing Formulas}

Now fix a variable $x$. We will build up two partial assignments, and will use the notation 
$$\mathcal{A}_1(x) = T \mbox{ and } \mathcal{A}_2(x) = F$$
to indicate that the first partial assignment sets $x$ to $T$, and the second one sets $x$ to $F$. Furthermore we will refer to the conditional distribution that is uniform on all satisfying assignments consistent with the decisions made so far in $\mathcal{A}_1$ and $\mathcal{D}_1$. Similarly we will refer to the other conditional distribution as $\mathcal{D}_2$. Note that these distributions are updated as more variables are set. 

We can now state our goal. Suppose we have partial assignments $\mathcal{A}_1$ and $\mathcal{A}_2$. Then we will want to write
$$\Phi_{\mathcal{A}_1} = \Phi_{I_1} \wedge \Phi_{O_1}$$
where $\Phi_{\mathcal{A}_1}$ is the subformula we get after making the assignments in $\mathcal{A}_1$ and simplifying \--- i.e. removing literals (a variable or its negation) that are $F$, and deleting clauses that already have a literal set to $T$. Similarly we will want to write 
$$\Phi_{\mathcal{A}_2} = \Phi_{I_2} \wedge \Phi_{O_2}$$
Finally, we want the following conditions to be met:
\begin{enumerate}

\item[$(1)$] $\Phi_{O_1} = \Phi_{O_2} (:= \Phi_O)$

\item[$(2)$]  $\Phi_{I_1}$ and $\Phi_O$ share no variables, and similarly for $\Phi_{I_2}$ and $\Phi_O$


\end{enumerate}

\noindent The crucial point is that if we can find partial assignments $\mathcal{A}_1$ and $\mathcal{A}_2$ where $\Phi_{\mathcal{A}_1}$ and $\Phi_{\mathcal{A}_2}$ meet the above conditions, then the conditional distribution on all variables in $\Phi_O$ is exactly the same. We will use the notation
$$\mathcal{D}_1 \Big |_{\mbox{vars}(\Phi_O)}$$
to denote the conditional distribution of $\mathcal{D}_1$ projected onto just the variables in $\Phi_O$. Then we have:

\begin{lemma}
If the above factorization conditions are met, then 
$$\mathcal{D}_1 \Big |_{\mbox{vars}(\Phi_O)} = \mathcal{D}_2 \Big |_{\mbox{vars}(\Phi_O)}$$
\end{lemma}

\begin{proof}
From the assumption that $\Phi_{\mathcal{A}_1} = \Phi_{I_1} \wedge \Phi_{O}$ and because $\Phi_{I_1}$ and $\Phi_O$ share no variables, it means that there are no clauses that contain variables from both the subformulas $\Phi_{I_1}$ and $\Phi_O$. Any such clause would prevent us from writing the formula $\Phi_{\mathcal{A}_1}$ in such a factorized form. 
Thus the distribution $\mathcal{D}_1$ is simply the cross product of the uniform distributions on satisfying assignments to $\Phi_{I_1}$ and $\Phi_O$. An identical statement holds for $\mathcal{D}_2$ which completes the proof. 
\end{proof}

\noindent Note that meeting the factorization conditions does {\em not} mean that the {\em number} of satisfying assignments to $\Phi_{\mathcal{A}_1}$ and $\Phi_{\mathcal{A}_2}$ are the same.

\subsection{Factorization via Coupling}

Our goal in this subsection is to give a coupling procedure to generate partial assignments $\mathcal{A}_1$ and $\mathcal{A}_2$ starting from $x = T$ and $x = F$ respectively, that result in a factorized formula. In fact, we will set exactly the same set $S$ of variables in both, although not all variables will be set to the same value in the two partial assignments and this set $S$ will also be random. 

There are two important constraints that we will impose on how we construct the partial assignments, that will make it somewhat tricky. First, suppose we have only set the variable $x$ and next we choose to set the variable $y$ in both $\mathcal{A}_1$ and $\mathcal{A}_2$. We will want that the distribution on how we set $y$ in the coupling procedure in $\mathcal{A}_1$ to match the conditional distribution $\mathcal{D}_1$ and similarly for $\mathcal{A}_2$. Now suppose we terminate with some set $S$ having been set. We can continue sampling the variables in $\bar{S}$ from $\mathcal{D}_1$, and we are now guaranteed that the full assignment we generate is uniform from the set of assignments with $x = T$. An identical statement holds when starting with $x = F$. Second, we will want that with very high probability, the coupling procedure terminates with not too many variables in the formula $\Phi_{I_1}$ or $\Phi_{I_2}$. Finally, we will assume that we are given access to a powerful oracle:

\begin{definition}
We will call the following a {\em conditional distribution oracle}: Given a CNF formula $\Phi$, a partial assignment $\mathcal{A}$ and a variable $y$ it can answer with the probability that $y = T$ in a uniformly random satisfying assignment that is also consistent with $\mathcal{A}$
\end{definition}

\noindent Such an oracle is obviously very powerful, and it is well known that if we had access to it we could compute the number of satisfying assignments to $\Phi$ exactly with a polynomial number of queries. However one should think of the coupling procedure as a thought experiment, which will be useful in an indirect way to build up towards our algorithm for approximate counting. 

\begin{fragment*}[t]
\caption{
\label{alg:iterrefine}{\sc Coupling Procedure}, \\ \textbf{Input:} Monotone CNF $\Phi$, variable $x$ and conditional distribution oracle $F$  \vspace*{0.01in}
}

\begin{enumerate} \itemsep 0pt
\small 
\item Using Lemma~\ref{lemma:marked}, label variables as marked or unmarked
\item Initialize $\mathcal{A}_1(x) = T$ and $\mathcal{A}_2(x) = F$
\item Initialize $V_I = \{x\}$ and $V_O = \{x_i\}_{i = 1}^n \setminus \{x\}$
\item While there is a clause $c$ with variables in both $V_I$ and $V_O$
\item $\quad$ Sequentially sample its marked variables (if any) from $\mathcal{D}_1$ and $\mathcal{D}_2$, using $F$ to construct best coupling at each step\footnotemark
\item $\quad$ Case \# 1: $c$ is satisfied by variables already set in both $\mathcal{A}_1$ and $\mathcal{A}_2$
\item $\quad$ $\quad$ Let $S$ be the variables in $c$ that have different truth values in $\mathcal{A}_1$ and $\mathcal{A}_2$. 
\item $\quad$ $\quad$ Update $V_I \leftarrow V_I \cup S$, $V_O \leftarrow V_O \setminus S$
\item $\quad$ $\quad$ Delete $c$
\item $\quad$ Case \# 2: $c$ is not satisfied by variables already set in either $\mathcal{A}_1$ or $\mathcal{A}_2$
\item $\quad$ $\quad$ Let $S$ be all variables in $c$ (marked or unmarked)
\item $\quad$ $\quad$ Update $V_I \leftarrow V_I \cup S$, $V_O \leftarrow V_O \setminus S$
\item End
\end{enumerate} 

\end{fragment*}

\footnotetext{Here by ``best coupling at each step" we mean that sequentially for each variable we want to maximize the probability that $\mbox{Pr}[\mathcal{A}_1(y) = \mathcal{A}_2(y)]$ while preserving the fact that $y$ is set in  $\mathcal{A}_1$ and $\mathcal{A}_2$ according to $\mathcal{D}_1$ and $\mathcal{D}_2$ respectively.}

Notice that a clause $c$ can only trigger the WHILE loop at most once. If it ends up in Case \# 1 then it is deleted from the formula. If it ends up in Case \# 2 then all its variables are included in $V_I$ and once a variable is included in $V_I$ it is never removed. Thus the procedure clearly terminates. Our first step is to show that when it does, the formula factorizes. 
Let $\mathcal{C}_I$ be the set of remaining clauses which have all of their variables in $V_I$. Similarly let $\mathcal{C}_O$ be the set of remaining clauses which have all of their variables in $V_O$. Then set
$$\Phi'_I = \wedge_{c \in \mathcal{C}_I} c$$
and let $\Phi_{I_1}$ and $\Phi_{I_2}$ be the simplification of $\Phi'_I$ with respect to the partial assignments $\mathcal{A}_1$ and $\mathcal{A}_2$. Similarly set
$$\Phi'_O = \wedge_{c \in \mathcal{C}_O} c$$
and let $\Phi_{O_1}$ and $\Phi_{O_2}$ be the simplification of $\Phi'_O$ with respect to the partial assignments $\mathcal{A}_1$ and $\mathcal{A}_2$.

\begin{claim}\label{claim:vi}
All variables with different truth assignments in $\mathcal{A}_1$ and $\mathcal{A}_2$ are in $V_I$. 
\end{claim}

\begin{proof}
A variable is set in response to it being contained in some clause $c$ that triggers the WHILE loop. Any such variable is moved into $V_I$ in both Case \# 1 and Case \# 2. 
\end{proof}

\noindent Now we have an immediate corollary that helps us towards proving that we have found partial assignments for which $\Phi$ factorizes:

\begin{corollary}
$\Phi_{O_1} = \Phi_{O_2}$
\end{corollary}

\begin{proof}
Recall that $\Phi_{O_1}$ and $\Phi_{O_2}$ come from simplifying $\Phi'_O$ (which contains only variables in $V_O$) according to $\mathcal{A}_1$ and $\mathcal{A}_2$. From Claim~\ref{claim:vi}, we know that  $\mathcal{A}_1$ and $\mathcal{A}_2$ are the same restricted to $V_O$ and thus we get the same formula in both cases. 
\end{proof}

\noindent Now that we know they are equal, we can define $\Phi_O = \Phi_{O_1} = \Phi_{O_2}$. What remains is to show that the subformulas we have are actually factorizations of the original formula $\Phi$:

\begin{lemma}\label{lemma:isfactored}
$\Phi_{\mathcal{A}_1} = \Phi_{I_1} \wedge \Phi_{O}$ and $\Phi_{\mathcal{A}_2} = \Phi_{I_2} \wedge \Phi_{O}$
\end{lemma}

\begin{proof}
When the WHILE loop terminates, every clause $c$ in the original formula $\Phi$ either has all of its variables in $V_I$ or in $V_O$, or was deleted because it already contains at least one variable in both $\mathcal{A}_1$ and $\mathcal{A}_2$ that satisfies it (although it need not be the same variable). Hence every clause in $\Phi$ that is not already satisfied in both $\mathcal{A}_1$ and $\mathcal{A}_2$ shows up in $\Phi'_I \wedge \Phi'_O$. Some clauses that are already satisfied in both may show up as well. In any case, this completes the proof because the remaining operation just simplifies the formulas according to the partial assignments. 
\end{proof}

\subsection{How Quickly Does the Coupling Procedure Terminate?}

What remains is to bound the probability that the number of variables included in $V_I$ is at most $t$. First we need an elementary definition:

\begin{definition}
When a variable $x_i$ is given different truth assignments in $\mathcal{A}_1$ and $\mathcal{A}_2$, we call it a {\em type $1$} error. When a clause $c$ has all of its marked variables set in both $\mathcal{A}_1$ and $\mathcal{A}_2$, but in at least one of them is not yet satisfied, we call it a {\em type $2$} error.
\end{definition}

\noindent Note that it is possible for a variable to participate in both a type $1$ and type $2$ error. In any case, these are the {\em only} reasons that a variable is included in $V_I$ in an execution of the coupling procedure:

\begin{observation}\label{obs:errors}
All variables in $V_I$ are included either due to a type $1$ error or a type $2$ error, or both. 
\end{observation}

Now our approach to showing that $V_I$ contains not too many variables with high probability is to show that if it did, there would be a large collection of disjoint errors. First we construct a useful graph underlying the process:

\begin{definition}
Let $G$ be the graph on vertices $V_I$ where we connect variables if and only if they appear in the same clause together ({\em any} clause from the original formula $\Phi$). 
\end{definition}

\noindent The crucial property is that it is connected:

\begin{observation}
$G$ is connected
\end{observation}

\begin{proof}
This property holds by induction. Assume that at the start of the WHILE loop, the property holds. Then at the end of the loop, any variable $x_i$ added to $V_I$ must have been contained in a clause $c$ that at the outset had one of its variables in $V_I$. This completes the proof. 
\end{proof}

Now by Observation~\ref{obs:errors}, for every variable in $V_I$ we can blame it on either a type $1$ or a type $2$ error. Both of these types of errors are unlikely. But for each variable, charging it to an error is problematic because of overlaps in the events. In particular, suppose we have two variables $x_i$ and $x_j$ that are both included in $V_I$. It could be that both variables are in the same clause $c$ which resulted in a type $2$ error, in which case we could only charge one of the variables to it. This turns out not to be a major issue. 

The more challenging type of overlap is when two clauses $c$ and $c'$ both experience type $2$ errors and overlap. In isolation, each clause would be unlikely to experience a type $2$ error. But it could be that $c$ and $c'$ share all but one of their marked variables, in which case once we know that $c$ experiences a type $2$ error, then $c'$ has a reasonable chance of experiencing one as well. We will get around this issue by building what we call a $3$-tree. This approach is inspired by Noga Alon's parallel algorithmic local lemma \cite{A} where he uses a $2,3$-tree. 

\begin{definition}
We call a graph $T$ on subset of $V_I$ a $3$-tree if each vertex is distance at least $3$ from all the others, and when we add edges between vertices at distance exactly $3$ the tree is connected. 
\end{definition}

Next we show that $G$ contains a large $3$-tree:

\begin{lemma}\label{lemma:3tree}
Suppose that any clause contains between $k$ and $6k$ variables. Then any maximal $3$-tree contains at least $\frac{|V_I|}{2(6dk)^2}$ vertices.
\end{lemma}

\begin{proof}
Consider a maximal $3$-tree $T$. We claim that every vertex $x_i \in V_I$ must be distance at most $2$ from some $x_j$ in $T$. If not, then we could take the shortest path from $x_i$ to $T$ and move along it, and at some point we would encounter a vertex that is also not in $T$ whose distance from $T$ is exactly $3$, at which point we could add it, contradicting $T$'s maximality. Now for every $x_i$ in $T$, we remove from consideration at most $(6dk)^2 + (6dk)$ other variables (all those at distance at most $2$ from $x_i$ in $G$). This completes the proof. 
\end{proof}

Now we can indeed charge every variable in $T$ to a disjoint error:

\begin{claim}\label{claim:disj}
If two variables $x_i$ and $x_j$ in $T$ are the result of type $2$ errors for $c$ and $c'$, then 
$$\mbox{vars}(c_i) \cap \mbox{vars}(c_j) = \emptyset$$
\end{claim}

\begin{proof}
For the sake of contradiction, suppose that $\mbox{vars}(c_i) \cap \mbox{vars}(c_j) \neq \emptyset$. Then since $c$ and $c'$ experience type $2$ errors, all of their variables are included in $V_I$. This gives a length $2$ path from $x_i$ to $x_j$ in $G$, which if they were both included in $T$, would contradict the assumption that $T$ is a $3$-tree. 
\end{proof}

We are now ready to prove the main theorem of this section:

\begin{theorem}\label{thm:couplingterminates}
Suppose that every clause contains between $k$ and $6k$ variables and that $\log d \geq 5 \log k + 20$ and $k \geq 50 \log d + 50 \log k + 250$.
Then $$\mbox{Pr}[|V_I| \geq 2 (6dk)^2 t] \leq \Big (\frac{1}{2} \Big )^t$$
\end{theorem}

\begin{proof}
First note that the conditions on $k$ and $d$ imply that the condition in Lemma~\ref{lemma:marked} holds. Now suppose that $|V_I| \geq 2 (6dk)^2 t$. Then by Lemma~\ref{lemma:3tree} we can find a $3$-tree $T$ with at least $t$ vertices. First we will work towards bounding the probability of any particular $3$-tree on $t$ vertices. We note that since each clause has at least $k/4$ marked variables and has at most $6k$ total variables, we can bound the probability of a type $1$ error as
$$\mbox{Pr}[x \mbox{ has type $1$ error}] \leq \frac{4}{d^{12}}$$
This uses the assumption $e d^{12} (6dk) \leq 2^{k/4}$ which allows us to choose $s = d^{12}$ in Corollary~\ref{corr:uniform}. Also we can bound the probability of a variable participating in a type $2$ error as
$$\mbox{Pr}[x \mbox{ participates in a type $2$ error}] \leq d\Big (\frac{49}{100} \Big)^{k/4} $$
which follows from Corollary~\ref{corr:uniform} using the condition that $s \geq d^3 \geq 100$ and that each variable belongs to at most $d$ clauses and each clause has at least $k/4$ marked variables. Now by Claim~\ref{claim:disj} we know that clauses that cause the type $2$ errors for each vertex in $T$ are disjoint. Thus putting it all together we can bound the probability of any particular $3$-tree on $t$ vertices as:
$$\Big( \frac{4}{d^{12}} + d \Big (\frac{49}{100} \Big)^{k/4} \Big)^t$$

Now it is well-known (see \cite{K, A}) that the number of trees of size $t$ in a graph of degree at most $\Delta$ is at most $(e\Delta)^t$. Moreover if we connect pairs of vertices in $G$ that are distance exactly $3$ from each other, then we get a new graph $H$ whose maximum degree is at most $\Delta = 3(6dk)^3$. Thus putting it all together we have that the probability that $|V_I| > 2 (6dk)^2 t$ can be bounded by 
$$ \Big( \frac{4}{d^{12}} + d \Big (\frac{49}{100} \Big )^{k/4} \Big)^t \Big ( 3 e (6dk)^3 \Big )^t \leq \Big( \frac{1}{2}  \Big)^t$$
where the last inequality follows from the assumptions that $\log d \geq 5 \log k + 20$ and $k \geq 50 \log d + 50 \log k + 250$. 
\end{proof}

Thus we can conclude that with high probability, the number of variables in $V_I$ is at most logarithmic. We can now brute-force search over all assignments to count the number of satisfying assignments to either $\Phi_{I_1}$ or $\Phi_{I_2}$. The trouble is that we do not have access to the marginal probabilities, so we cannot actually execute the coupling procedure. We will need to circumvent this issue next.

\section{Implications of the Coupling Procedure}

In this section, we give an abstraction that allows us to think about the coupling procedure as a randomly chosen root-to-leaf path in a certain tree whose nodes represent states. First, we make an elementary observation that will be useful in discussing how this tree is constructed. Recall that the coupling procedure chooses {\em any} clause that contains variables in both $V_I$ and $V_O$ and then samples {\em all} marked variables in it. We will assume without loss of generality that the choices it makes are done in lexicographic order. So if the clauses in $\Phi$ are ordered arbitrarily as $c_1, c_2, \dots, c_m$ and the variables are ordered as $x_1, x_2, \dots, x_n$ when executing the WHILE loop, if it has a choice of more than one clause it chooses among them the clause $c_i$ with the lowest subscript $i$. Similarly, given a choice of which marked variable to sample next, it chooses among them the $x_j$ with the lowest subscript $j$. 

The important point is that now we can think of a state associated with the coupling procedure, which we will denote by $\sigma$. 

\begin{definition}\label{def:states}
The state $\sigma$ of the coupling procedure specifies the following:
\begin{enumerate}

\item The set of remaining clauses $\mathcal{C}'$ \--- i.e. that have not yet been deleted

\item The partition of the variables into $V_I$ and $V_O$

\item The set $S$ of variables whose values have been set, along with their values in both $\mathcal{A}_1$ and $\mathcal{A}_2$

\item The current clause $c^*$ being operated on in the while loop, if any

\end{enumerate}
\end{definition}

\noindent We will assume that the set $\mathcal{M}$ of marked variables is fixed once and for all. Now the transition rules are that if $c^*$ has any marked variables that are unset, it chooses the lexicographically first and sets it. And when $c^*$ has no remaining marked variables to set, it updates $\mathcal{C}'$, $V_I$ and $V_O$ according to whether it falls into Case \# 1 or Case \#2 and sets the current clause to empty. Finally, if the current clause is empty then it chooses the lexicographically first clause from $\mathcal{C}'$ which has at least one variable in each of $V_I$ and $V_O$ to be $c^*$. 

Finally, we can define the next variable operation:

\begin{definition}\label{def:nextset}
Let $\mathcal{R}: \Sigma \rightarrow \{x_i\}_{i =1}^n \cup \{\emptyset\} \times \Sigma$ be the function that takes in a state $\sigma$, transitions to the next state $\sigma'$ that sets some variable $y$ and outputs $(y, \sigma')$. 
\end{definition}

Note that some states $\sigma$ do not immediately set a variable \--- e.g.  if the next operation is to choose the next clause, or update $\mathcal{C}'$, $V_I$ and $V_O$. These latter transitions are deterministic, so we let $\sigma'$ be the end resulting state and $y$ be the variable that it sets. Now we can define the stochastic decision tree underlying the coupling procedure:

\begin{definition}\label{def:twosided}
Given a conditional distribution oracle $F$, the function $\mathcal{R}$ and a stopping threshold $s$, the associated {\em stochastic decision tree} is the following:
\begin{enumerate}

\item[(1)] The root node corresponds to the state where only $x$ is set, $\mathcal{A}_1(x) = T$, $\mathcal{A}_2(x) = F$, $V_I = \{x\}$ and $V_O = \{x_i\}_{i = 1}^n \setminus \{x\}$. 

\item[(2)] Each node has either zero or four descendants. If the current node corresponds to state $\sigma$, let $(y, \sigma') = \mathcal{R}(\sigma)$. Then if $y = \emptyset$ or if $|V_I| = \tau$ there are no descendants and the current node is a leaf corresponding to the termination of the coupling procedure or $|V_I|$ being too large. Otherwise the four descendants correspond to the four choices for how to set $y$ in $\mathcal{A}_1$ and $\mathcal{A}_2$, and are marked with the state $\sigma''$ which incorporates their respective choices into $\sigma'$. 

\item[(3)] Moreover the probability on an edge from a state $\sigma'$ to a state $\sigma''$ where $y$ has been set as $\mathcal{A}_1(y) = T$ and $\mathcal{A}_2(y) = T$ is equal to 
$$\min(\mathcal{D}_1(y), \mathcal{D}_2(y))$$
and the transition to the state where $\mathcal{A}_1(y) = F$ and $\mathcal{A}_2(y) = F$ has probability
$$\min(1-\mathcal{D}_1(y), 1-\mathcal{D}_2(y))$$
Finally if $\mathcal{D}_1(y) > \mathcal{D}_2(y)$ then the transition to $\mathcal{A}_1(y) = T$ and $\mathcal{A}_2(y) = F$ is non-zero and is assigned all the remaining probability. Otherwise the transition to $\mathcal{A}_1(y) = F$ and $\mathcal{A}_2(y) = T$ is non-zero and is assigned all the remaining probability.

\end{enumerate}
\end{definition}

Now we can use the stochastic decision tree to give an alternative procedure to sample a uniformly random satisfying assignment of $\Phi$. We will refer to the process of starting from the root, and choosing a descendant with the corresponding transition probability, until a leaf node is reached as ``choosing a random root-to-leaf path". 

\begin{fragment*}[t]
\caption{
\label{alg:iterrefine}{\sc Decision Tree Sampling}, \\ \textbf{Input:} Monotone CNF $\Phi$, stochastic decision tree $S$  \vspace*{0.01in}
}

\begin{enumerate} \itemsep 0pt
\small 
\item Choose a random root-to-leaf path in $S$
\item Choose a uniformly random assignment $A_1$ consistent with $\mathcal{A}_1$
\item Choose a uniformly random assignment $A_2$ consistent with $\mathcal{A}_2$
\item Output $A_1$ with probability $q = \mbox{Pr}_{\mathcal{D}}[x = T]$, and otherwise output $A_2$

\end{enumerate} 

\end{fragment*}

\begin{claim}\label{lemma:altgen}
The decision tree sampling procedure outputs a uniformly random satisfying assignment of $\Phi$. 
\end{claim}

\begin{proof}
We could alternatively think of the decision tree sampling procedure as deciding on whether (at the end) to output $A_1$ or $A_2$ with probability $q$ vs. $1-q$ at the outset. Then if we choose to output $A_1$, and we only keep track of the choices made for $\mathcal{A}_1$, marginally these correspond to sequentially sampling the assignment of variables from $\mathcal{D}_1$. And when we reach a leaf node in $S$ we can interpret the remaining choices to $A_1$ as sampling all unset variables from $\mathcal{D}_1$. Thus the output in this case is a uniformly random satisfying assignment with $x = T$. An identical statement holds for when we choose to output $A_2$, and because we decided between them at the outset with the correct probability, this completes the proof of the claim. 
\end{proof}

Now let $\sigma$ be the state of a leaf node $u$ and let $\mathcal{A}_1$ and $\mathcal{A}_2$ be the resulting partial assignments. Let $p_1$ be the product of certain probabilities along the root-to-leaf path. In particular, suppose along the path there is a transition with $y$ being set. Let $q_1$ be the probability of the transition to $(\mathcal{A}_1(y), \mathcal{A}_2(y))$ \--- i.e. along the branch that it actually went down. And let $q_2$ be the probability of the transition to $(\mathcal{A}_1(y), \overline{\mathcal{A}_2(y)})$ \--- i.e. where $y$ is set the same in $\mathcal{A}_1$ but is set to the opposite value as it was in $\mathcal{A}_2$. We let $p_1$ be the product of all $\frac{q_1}{q_1 + q_2}$ over all such decision on the root-to-leaf path. 

\begin{lemma}\label{lemma:paths}
Let $A$ be an assignment that agrees with $\mathcal{A}_1$. Then for the {\sc Decision Tree Sampling} procedure
$$\mbox{Pr}\Big [ \mbox{terminates at leaf $u$} \Big | \mbox{outputs assignment $A$}\Big ] = p_1$$
\end{lemma}

\begin{proof}
The idea behind this proof is to think of the random choice of which of the four descendants to transition to as being broken down into two separate random choices where we first choose $\mathcal{A}_1(y)$ and then we choose $\mathcal{A}_2(y)$. See Figure~\ref{fig:releveling}.
Now we can make the random choices in the {\sc Decision Tree Sampling} procedure in an entirely different order. Instead of choosing the transition in the first layer, then the second layer and so on, we instead make {\em all} of the choices in the odd layers. Moreover at each leaf, we choose which assignment consistent with $\mathcal{A}_1$ we would output. This is the first phase. Next we choose whether to output the assignment consistent with $\mathcal{A}_1$ or with $\mathcal{A}_2$. Finally, we make all the choices in the even layers which fixes the root-to-leaf path and then we choose an assignment consistent with $\mathcal{A}_2$. This is the second phase. 

The key point is that once the output $A$ is fixed, all of the choices in the first phase are determined, because every time a variable $y$ is set it must agree with its setting in $A$. Moreover each leaf node must choose $A$ for its assignment consistent with $\mathcal{A}_1$. And finally, we know that the sampling procedure must output the assignment consistent with $\mathcal{A}_1$ because $A$ agrees with $\mathcal{A}_1$ and not $\mathcal{A}_2$ (because they differ on how they set $x$). Thus conditioned on outputting $A$ the only random choices left are those in the second phase. Now the lemma follows because the probability of reaching leaf node $u$ is exactly the probability along the path of all of the even layer choices, which is how we defined $p_1$. 
\end{proof}

\begin{figure}
\begin{center}
    \includegraphics[width=12cm]{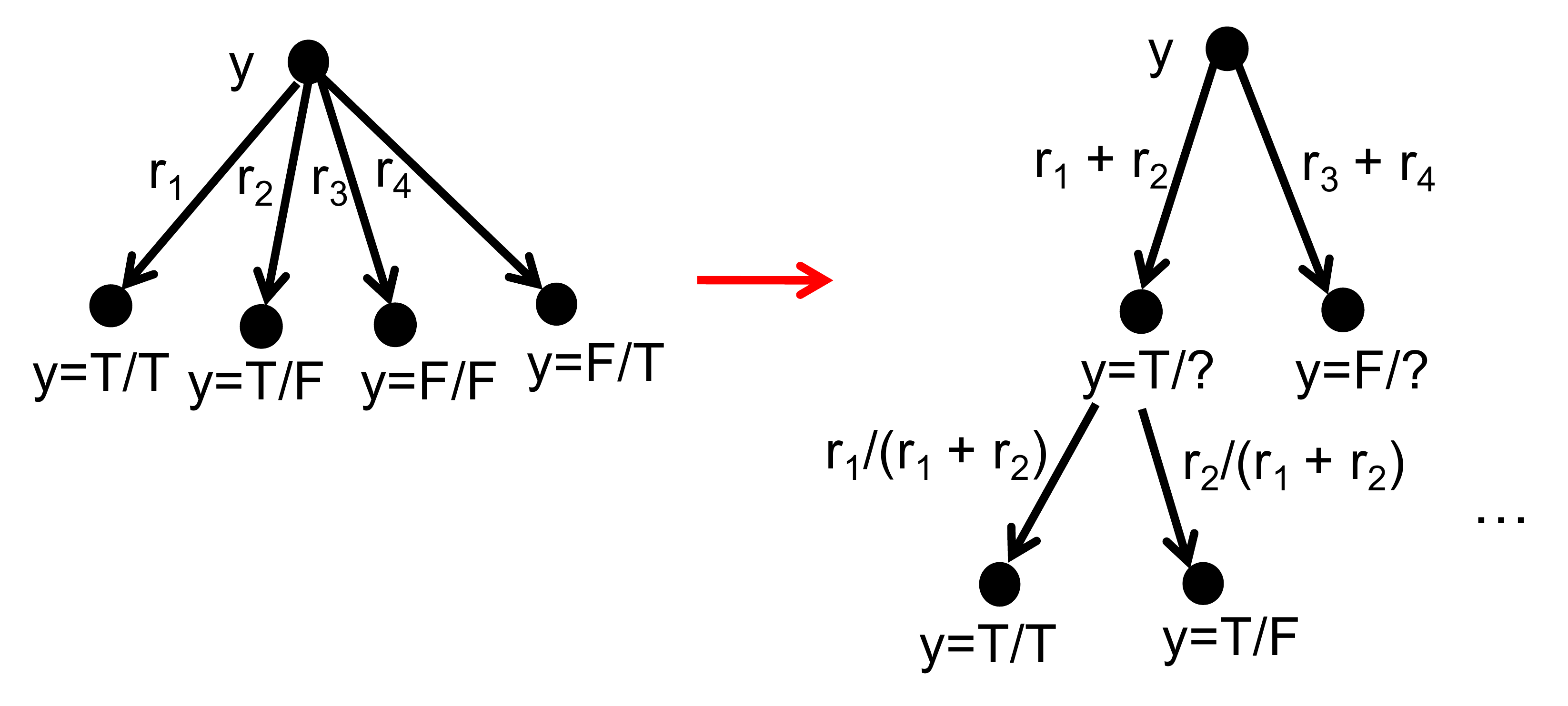}
    \caption{
      A transformation on the stochastic decision tree that makes it easier to understand what happens when we condition on the assignment $A$ that is output by {\sc Decision Tree Sampling}. 
    } \label{fig:releveling}
  \end{center}
\end{figure}

We can define $p_2$ in an analogous way to how we defined $p_1$ (i.e. as the product of certain probabilities along the root-to-leaf path), and the lemma above shows that $p_2$ is exactly the  probability of all the decisions made along the root-to-leaf path conditioned on the output being $A$ where $A$ agrees with $\mathcal{A}_2$. 

The key lemma is the following:

\begin{lemma}\label{lemma:balance}
Let $N_1$ be the number of satisfying assignments consistent with $\mathcal{A}_1$ and let $N_2$ be the number of satisfying assignments consistent with $\mathcal{A}_2$. Then 
$$\frac{p_1 N_1}{p_2 N_2} = \frac{q}{1-q}$$
\end{lemma}

\begin{proof}
Let $u$ be a leaf node. Consider a random variable $Z_u$ that when we run the decision tree sampling procedure is non-zero if and only if we end at $u$. Moreover let $Z_u = (1-q)$ if an assignment with $x = T$ is output, and $Z_u = -q$ if an assignment with $x = F$ is output. Then clearly $\E[Z_u] = 0$. Now alternatively we can write:
$$\E[Z_u] = \E_A[ \E[Z_u | A \mbox{ is output}]]$$
where $A$ is a uniformly random satisfying assignment of $\Phi$, precisely because of Lemma~\ref{lemma:altgen}. Let $N$ be the total number of such assignments. Then
$$\E_A[ \E[Z_u | A]] = \Big (\frac{N_1}{N} \Big ) (p_1) (1-q) + \Big (\frac{N_2}{N} \Big ) (p_2) (-q)$$
This follows because the only assignments $A$ that can be output at $u$ must be consistent with either $\mathcal{A}_1$ or $\mathcal{A}_2$. Note that these are disjoint events because in one of them $x = T$ while in the other $x = F$. Then once we know that $A$ is consistent with $\mathcal{A}_1$ (which happens with probability $\frac{N_1}{N}$) the probability for the decisions made in $\mathcal{A}_2$ being such that we reach $u$ is exactly $p_1$, as this was how it was defined. The final term in the product of three terms is just the value of $Z_u$. An identical argument justifies the second term. Now using the fact that the above expression evaluates to zero and rearranging completes the proof. 
\end{proof}

\section{Certifying the Marginal Distribution}

\subsection{One-Sided Stochastic Decision Trees}

The stochastic decision tree that we defined in the previous section is a natural representation of the trajectory of the coupling procedure. However it has an important drawback that we will remedy here. Its crucial property is captured in Lemma~\ref{lemma:balance} which gives a relation between 
\begin{enumerate}

\item[(1)] $p_i$ \--- the conditional probability of an assignment consistent with $\mathcal{A}_i$ reaching $u$ and

\item[(2)] $N_i$ \--- the number of assignments consistent with $\mathcal{A}_i$

\end{enumerate}

\noindent for $i = 1, 2$. However $p_1$ is the product of various ratios of probabilities along the root-to-leaf path. This means that if we think of the transition probabilities as variables, the constraint imposed by Lemma~\ref{lemma:balance} is far from linear\footnote{What's worse is that the contribution of a particular decision to $p_1$ and $p_2$ is a multiplication by one of two ratios of probabilities, {\em which have different denominators}. For reasons that we will not digress into, this makes it challenging to encode the total probability $p_1$ as a flow in a tree.}.

In this section, we will transform a stochastic decision tree into two separate trees, that we call {\em one-sided stochastic decision trees}. These will have the property that the constraint imposed by Lemma~\ref{lemma:balance} will be linear in the unknown probabilities that we think of as variables. Ultimately we will show that any such pair can $(1)$ certify that a given value $q$ is within an additive inverse polynomial factor of $\mbox{Pr}_{\mathcal{D}}[x = T]$ and $(2)$ can be constructed in polynomial time through linear programming. First we explain the transformation from a stochastic decision tree to a one-sided stochastic decision tree. We will then formally define its properties and what we require of it. 

Now suppose we are given a stochastic decision tree $S$. Let's construct the one-sided stochastic decision tree $S_1$ that represents the trajectory of the partial assignment $\mathcal{A}_1$. When we start from the starting state $\sigma$ (see Definition~\ref{def:states}), the four descendants of it in $S$ will now be four grand-children. Its immediate descendants will be two nodes $u$ and $u'$, one representing the choice $\mathcal{A}_1(y) = T$ and one representing $\mathcal{A}_1(y) = F$, where $y$ is the next variable set (see Definition~\ref{def:nextset}). The two children of $\sigma$ in $S$ that correspond to $\mathcal{A}_1(y) = T$ will now be the children of $u$ and the other two children will now be the children of $u'$. We will continue in this way so that alternate layers represent nodes present in $S$ and new nodes. 

This alone does not change much the semantics of the trajectory. All we are doing is breaking up the decision of which of the four children to proceed to, into two separate decisions. The first decision is based on just $\mathcal{A}_1$ and the second is based on $\mathcal{A}_2$. However we will change the semantics of what probabilities we associate with different transitions. For starters, we will work with total probabilities. So the total probability incoming into the starting node is $1$. Let's see how this works inductively. Let's now suppose that $\sigma$ represents the state of some node in $S$ (not necessarily the starting state) and $u$ and $u'$ are its descendants in $S_1$. Then if the total probability into $\sigma$ in $S_1$ is $z$, we place $z$ along both the edges to $u$ and to $u'$. This is because the decision tree is now from the perspective of $\mathcal{A}_1$, who perhaps has already chosen his assignment uniformly at random from the satisfying assignments with $x = T$ but has not set all of those values in $\mathcal{A}_1$. Hence his decision is not a random variable, since given the option of transition to $u$ or $u'$ he must go to whichever one is consistent with his hidden values. 

However from this perspective, the choices corresponding to $\mathcal{A}_2$ are random because he has no knowledge of the assignment that the other player is working with. If we have $z$ total probability coming into $u$, then the total probability into its two descendants will be $(\frac{q_1}{q_1 + q_2}) z$ and  $(\frac{q_2}{q_1 + q_2}) z$ respectively, where $q_1$ and $q_2$ were the probabilities on the transitions in $S$ into the two corresponding descendants. 
In particular, if $q_1$ is the probability of setting $\mathcal{A}_1(y) = T$ and $\mathcal{A}_2(y) = T$ and $q_2$ is the probability of setting $\mathcal{A}_1(y) = T$ and $\mathcal{A}_2(y) = F$ then $(\frac{q_1}{q_1 + q_2}) z$ is the total probability on the transition from $u$ to the descendant where $\mathcal{A}_2(y) = T$ and $(\frac{q_2}{q_1 + q_2}) z$ is the total probability on the transition from $u$ to the descendant where $\mathcal{A}_2(y) = F$. Note that from Corollary~\ref{corr:uniform} we have that 
$$\Big (\frac{q_2}{q_1 + q_2} \Big ) z \leq  \Big (\frac{4}{s} \Big ) z$$
This is an important property that we will make crucial use of later. Notice that it is a linear constraint in the total probability. Now we are ready to define a one-sided stochastic decision tree, which closely mirrors Definition~\ref{def:twosided}.

\begin{definition}\label{def:onesided}
Given the function $\mathcal{R}$ and a stopping threshold $\tau$, the associated {\em one-sided stochastic decision tree} for $\mathcal{A}_1$ is the following:
\begin{enumerate}

\item[(1)] The root node corresponds to the state where only $x$ is set, $\mathcal{A}_1(x) = T$, $\mathcal{A}_2(x) = F$, $V_I = \{x\}$ and $V_O = \{x_i\}_{i = 1}^n \setminus \{x\}$. 

\item[(2)] Each node has either two descendants and four {\em grand}-descendants or zero descendants. If the current node $a$ corresponds to state $\sigma$, let $(y, \sigma') = \mathcal{R}(\sigma)$. Then if $y = \emptyset$ or if $|V_I| = \tau$ there are no descendants and the current node is a leaf corresponding to the termination of the coupling procedure or $|V_I|$ being too large. Otherwise the two descendants corresponds to the two choices for how to set $y$ in $\mathcal{A}_1$. Each of their two descendants correspond to the two choices for how to set $y$ in $\mathcal{A}_2$. Each grand-descendant is marked with the state $\sigma'$ which incorporates their respective choices. 

\item[(3)] Let $z$ be the total probability into $a$. Then the total probability into each descendant is $z$. Moreover let the total probability into the grand-descendants with states $\mathcal{A}_1(y) = T$ and $\mathcal{A}_2(y) = T$ and $\mathcal{A}_1(y) = T$ and $\mathcal{A}_2(y) = F$ be $z_1$ and $z_2$ respectively. Then $z_1$ and $z_2$ are nonnegative, sum to $z$ and satisfy
$z_2 \leq  (\frac{4}{s}  ) z$. Similarly, let the total probability into the grand-descendants with states $\mathcal{A}_1(y) = F$ and $\mathcal{A}_2(y) = F$ and $\mathcal{A}_1(y) = F$ and $\mathcal{A}_2(y) = T$ be $z_3$ and $z_4$ respectively. Then $z_3$ and $z_4$ are nonnegative, sum to $z$ and satisfy
$z_4 \leq  (\frac{4}{s}  ) z$.

\end{enumerate}
\end{definition}

The one-sided stochastic decision tree for $\mathcal{A}_2$ is defined analogously, in the obvious way. Finally we record an elementary fact:

\begin{claim}\label{claim:match}
There is a perfect matching between the root-to-leaf paths in $S_1$ and $S_2$, so that any pair of assignments $A_1$ and $A_2$ that takes a root-to-leaf path $p$ in $S_1$, must also take the root-to-leaf path in $S_2$ to which $p$ is matched.
\end{claim}

\begin{proof}
Recall that the odd levels in $S_1$ and $S_2$ correspond to the nodes in $S$. Therefore from a root-to-leaf path $p$ in $S_1$ we can construct the root-to-leaf path in $S$, which in turn uniquely defines a root-to-leaf path in $S_2$ (because it specifies which nodes are visited in odd layers, and all paths end on a node in an odd layer). 
\end{proof}

\subsection{An Algorithm for Finding a Valid $S_1$ and $S_2$}

We are now ready to prove one of the two main theorems of this section:

\begin{theorem}\label{thm:constructive}
Let $q = \mbox{Pr}_{\mathcal{D}}[x = T]$ and $q' \leq q \leq q''$. Then there are two one-sided stochastic decision trees $S_1$ and $S_2$ that for any pair of matched root-to-leaf paths terminating in $u$ and $u'$ respectively satisfy
$$\Big ( \frac{q'}{1-q'} \Big ) p_2 N_2 \leq p_1 N_1 \leq \Big ( \frac{q''}{1-q''} \Big ) p_2 N_2$$
where $N_1$ and $N_2$ are number of satisfying assignments consistent with $\mathcal{A}_1$ and $\mathcal{A}_2$ respectively, and $p_1$ and $p_2$ are the total probability into $u$ and $u'$ respectively. 

Moreover given $q'$ and $q''$ that satisfy $q' \leq q \leq q''$ there is an algorithm to construct two one-sided stochastic decision trees $S_1$ and $S_2$ that satisfy the above condition on all matched leaf nodes corresponding to a termination of the coupling procedure, which runs in time polynomial in $m$ and $4^\tau$ where $\tau$ is the stopping size. 
\end{theorem}

\begin{proof}
The first part of the theorem follows from the transformation we gave from a stochastic decision tree to two one-sided stochastic decision trees. Then Claim~\ref{claim:match} combined with Lemma~\ref{lemma:balance} implies $\frac{q}{1-q} = \frac{p_1 N_1}{p_2 N_2}$, which then necessarily satisfies $\frac{q'}{1-q'} \leq \frac{p_1 N_1}{p_2 N_2} \leq \frac{q''}{1-q''}$. Rearranging completes the proof of the first part. 

To prove the second part of the theorem, notice that if $\tau$ is the stopping size, then the number of leaf nodes in $S_1$ and in $S_2$ is bounded by $4^\tau$. At each leaf node that corresponds to a termination of the coupling procedure, from Lemma~\ref{lemma:isfactored} we can compute the ratio of $N_1$ to $N_2$ as the ratio of the number of satisfying assignments to $\Phi_{I_1}$ to the number of satisfying assignments to $\Phi_{I_2}$. This can be done in polynomial in $m$ and $2^\tau$ time by brute-force. Finally, the constraints in Definition~\ref{def:onesided} are all linear in the variables that represent total probability (if we treat $\frac{4}{s}$, $\frac{q'}{1-q'}$, $\frac{q''}{1-q''}$ and all ratios $\frac{N_1}{N_2}$ as given constants). Thus we can find a valid choice of the total probability variables by linear programming. This completes the proof of the second part. 
\end{proof}

\noindent Recall that we will be able to choose $\tau = 2 c (6dk)^2 \log n$ and Theorem~\ref{thm:couplingterminates} will imply that at most a $1/n^c$ fraction of the distribution fails to couple. Thus the algorithm above runs in polynomial time for any constants $d$ and $k$. What remains is to show that any valid choice of total probabilities {\em certifies} that $q' \leq \mbox{Pr}_{\mathcal{D}}[x = T] \leq q''$.

\subsection{A Fractional Matching to Certify $q$}

We are now ready to prove the second main theorem of this section. We will show that having any two one-sided stochastic decision trees that meet the constraints on the leaves imposed by Theorem~\ref{thm:constructive} is enough to certify that $\mbox{Pr}_{\mathcal{D}}[x = T]$ is approximately between $q'$ and $q''$. This result will rest on two facts. Fix any assignment $A$. Then either
\begin{enumerate}

\item[(1)] The assignment has too many clauses that restricted to marked variables are all $F$ or

\item[(2)] The total probability of $A$ reaching a leaf node $u$ where the coupling procedure failed to terminate before reaching size $\tau$ is at most $O(\frac{1}{n^c})$. 

\end{enumerate}

\begin{theorem}\label{thm:certify}
Suppose that every clause contains between $k$ and $6k$ variables and that $\log d \geq 5 \log k + 20$ and $k \geq 50 \log d + 50 \log k + 250$. Then any two one-sided stochastic decision trees $S_1$ and $S_2$ that meet the constraints on the leaves imposed by Theorem~\ref{thm:constructive} and satisfy $\tau =  20 c (6dk)^2 \log n$ imply that $$q' - O\Big(\frac{1}{n^c}\Big)\leq \mbox{Pr}_{\mathcal{D}}[x = T] \leq q'' + O\Big(\frac{1}{n^c}\Big)$$
\end{theorem}

\noindent The proof of this theorem will use many of the same tools that appeared in the proof of Theorem~\ref{thm:couplingterminates}, since in essence we are performing a one-sided charging argument. 

\begin{proof}
The proof will proceed by constructing a complete bipartite graph $H = (U, V, E)$ and finding a fractional approximate matching as follows. The nodes in $U$ represent the satisfying assignments of $\Phi$ with $x = T$. The nodes in $V$ represent the satisfying assignments of $\Phi$ with $x = F$. Moreover all but a $O(\frac{1}{n^c})$ fraction of the nodes on the left will send between $1-q'' - O(\frac{1}{n^c})$ and $1-q' + O(\frac{1}{n^c})$ flow along their outgoing edges. Finally all but a $O(\frac{1}{n^c})$ fraction of the nodes on the right will receive between $q' - O(\frac{1}{n^c})$ and $q'' + O(\frac{1}{n^c})$ flow along their incoming edges.

First notice that any assignment $A$ (say with $x = T$) is mapped by $S_1$ to a distribution over leaf nodes, some of which correspond to a coupling and some of which correspond to a failure to couple before reaching size $\tau$. Now consider matched pairs of leaf nodes (according to Claim~\ref{claim:match}) that correspond to a coupling. Let $p_1$ and $p_2$ be the total probability of the leaf nodes in $S_1$ and $S_2$ respectively. Let $N_1$ and $N_2$ be the total number of assignments that are consistent with $\mathcal{A}_1$ and $\mathcal{A}_2$, and let $\mathcal{N}_1$ and $\mathcal{N}_2$ be the corresponding sets of assignments. From the assumption that $$\Big ( \frac{q'}{1-q'} \Big ) p_2 N_2 \leq p_1 N_1 \leq \Big ( \frac{q''}{1-q''} \Big ) p_2 N_2$$
and the intermediate value theorem it follows that there is a $q' \leq q^* \leq q''$ which satisfies
$$\Big ( \frac{q'^*}{1-q^*} \Big ) p_2 N_2 = p_1 N_1$$
Hence there is a flow that sends exactly $(1-q^*)p_1$ units of flow out of each node in $\mathcal{N}_1$ and which each node in $\mathcal{N}_2$ receives exactly $q^* p_2$ units of flow. 

If every leaf node corresponded to a coupling, we would indeed have the fractional matching we are looking for, just by summing these flows over all leaf nodes. What remains is to handle the leaf nodes that do not correspond to the coupling terminating before size $s$. Consider any such leaf node $u$ in $S_1$ and the corresponding leaf node $v$ in $S_2$. From Lemma~\ref{lemma:3tree} we have that there is a $3$-tree $T$ of size at least $10 c \log n$. For each node in $T$, from Claim~\ref{claim:disj} we have there are at least $10 c \log n$ disjoint type $1$ or type $2$ errors. 

\textbf{Case \# 1:} Suppose that there are at least $3.5 c \log n$ disjoint type $1$ errors. Fix the $3$-tree $T$, and look at all root-to-leaf paths that are consistent with just the type $1$ errors. Then the sum of their total probabilities is at most 
$$\Big ( \frac{4}{d^{12}} \Big )^{3.5 c \log n}$$
This follows because the constraint that $z_2 \leq  (\frac{4}{s}  ) z$ (and similarly for $z_4$) in Definition~\ref{def:onesided} implies that for each path we can factor out the above term corresponding to just the decisions where there are type $1$ errors. Moreover we chose $s = d^{12}$ exactly as we did in the proof of Theorem~\ref{thm:couplingterminates}. The remaining probabilities are conditional distributions on the paths (after having taken into account the type $1$ errors) and sum to at most one. Finally the total number of $3$-trees of size $10 c \log n$ is at most $(3e (6dk)^3)^{10 c \log n}$. Thus for any assignment $A$, if we ignore what happens to it when it ends up at a leaf node which did not couple and which has at least $3.5 c \log n$ disjoint type $1$ errors, in total we have ignored at most $$\Big ( \frac{4}{d^{12}} \Big )^{3.5 c \log n} \Big ( 3e (6dk)^3 \Big )^{10 c \log n} \leq 1/n^c$$ of its probability, where the last inequality uses the fact that $\log d \geq 5 \log k + 20$. 

\textbf{Case \# 2:} Suppose that there are at least $6.5 c \log n$ disjoint type $2$ errors. Each type $2$ error can be blamed on either $\mathcal{A}_1$ or $\mathcal{A}_2$ or both (e.g. it could be that the clause $c$ might only have all of its marked variables set to $F$ in $\mathcal{A}_1$). Let's suppose that the assignment $A$ contributes at least $2.5 c \log n$ disjoint type $2$ errors. In this case we will completely ignore $A$ in the constraints imposed by our flow. How many such assignments can there be? The probability of getting any such assignment is bounded by $$\Big (  \Big (\frac{49}{100} \Big )^{k/4} \Big )^{2.5 c \log n} \Big ( 3e (6dk)^3 \Big )^{10 c \log n}  \leq 1/n^c$$
where the last inequality has used the fact that $k \geq 50 \log d + 50 \log k + 250$.

Thus if we ignore the flow constraints for all such assignments, we will be ignoring at most a $1/n^c$ fraction of the nodes in $U$ and the nodes in $V$. The only remaining case is when the assignment $A$ ends up at a leaf node $u$ that has at least $6.5 c \log n$ disjoint type $2$ errors, but it contributes less than $2.5 c \log n$ itself. For each type $2$ error that it does not contribute to, it contributes to another type $1$ error. The only minor complication is that the node responsible might not be in the $3$-tree $T$. However it is distance at most $1$ from the $3$-tree because it is contained in a clause that results in type $2$ error that does contain a node in $T$. Now by an analogous reasoning as in Case \#1 above, if we fix the pattern of these type $1$ errors \--- i.e. we fix the $3$-tree and the extra nodes at distance $1$ from it that contribute the missing type $1$ errors \--- the sum of the total probability of all consistent root-to-leaf paths is at most 
$$\Big ( \frac{4}{d^{12}} \Big )^{4 c \log n} $$
Now the number of patterns can be bounded by $(4 e (6kd)^4 )^{10 c \log n}$, which accounts for the inclusion of extra nodes that are not in $T$. Once again, for such an assignment $A$ if we ignore what happens to it when it ends up at a leaf node which did not couple and which has at least $6.5 c \log n$ disjoint type $2$ but it contributes less than $2.5 c \log n$ itself, in total 
we have ignored at most $$\Big ( \frac{4}{d^{12}} \Big )^{4 c \log n} \Big (4 e (6kd)^4 \Big )^{10 c \log n} \leq 1/n^c$$ of its probability, where the last inequality uses the fact that $\log d \geq 5 \log k + 20$. 

Now returning to the beginning of the proof and letting $N_1$ and $N_2$ be the total number of satisfying assignments with $x = T$ and $x = F$ respectively. We have that the flow in the bipartite graph implies
$$(1-q'')N_1 - O\Big(\frac{1}{n^c}\Big)\leq \mbox{flowout}_U = \mbox{flowin}_V \leq q'' N_2 + O\Big(\frac{1}{n^c}\Big)$$
and the further condition
$$q' N_2 - O\Big(\frac{1}{n^c}\Big)\leq \mbox{flowin}_V = \mbox{flowout}_U \leq (1-q') N_1 + O\Big(\frac{1}{n^c}\Big)$$
which gives
$\frac{q'}{1-q'} - O(\frac{1}{n^c}) \leq \frac{q}{1-q} \leq \frac{q''}{1-q''} + O(\frac{1}{n^c})$
which completes the proof of the theorem. 
\end{proof}

\section{Applications}

Here we show how to use our algorithm for computing marginal probabilities when $k$ is logarithmic in $d$ for approximate counting and sampling from the uniform distribution on satisfying assignments. Throughout this section we will assume that each clause contains between $k$ and $2k$ variables and make use of Theorem~\ref{thm:certify} as a subroutine, which makes a weaker assumption that the number of variables per clause is between $k$ and $6k$. 

\subsection{Approximate Counting}

There is a standard approach for how to use an algorithm for computing marginal probabilities to do approximate counting in a monotone CNF, where no variable is negated (see e.g. \cite{BGGGS}). Essentially, we fix an ordering of the variables $x_1, x_2, \dots, x_n$ and a sequence of formulas $\Phi_1, \Phi_2, \dots, \Phi_n$. Let $\Phi_1= \Phi$ and let $\Phi_{i}$ be the subformula we get when substituting $x_1 = T, x_2 = T, \dots, x_{i-1} = T$ into $\Phi$ and simplifying. Notice that each such formula is a monotone CNF and inherits the properties we need from $\Phi$. In particular, each clause has at least $k$ variables because the only clauses left in $\Phi_i$ (i.e. not already satisfied) are the ones which have all of their variables unset. 

However such approaches crucially use monotonicity to ensure that no clause becomes too small (i.e. contains few variables, but is still unsatisfied). This is a similarly issue to what happened with the coupling procedure, which necessitating using {\em marked} and {\em unmarked} variables, the latter being variables that are never set and are used to make sure no clause becomes too small. We can take a similar approach here. In what follows we will no longer assume $\Phi$ is monotone. 

\begin{lemma}\label{lemma:nonmonotonehelper}
Suppose that $e (D+1)  \leq (32/31)^k$. Then there is a partial assignment $\mathcal{A}$ so that every clause is satisfied and each clause has at least $7k/8$ unset variables. Moreover there is a randomized algorithm to find such a partial assignment that runs in time polynomial in $m$, $n$, $k$ and $d$. Alternatively there is a deterministic algorithm that runs in time polynomial in $m$ and $n^{O(d^2)}$. 
\end{lemma}

\begin{proof}
We will choose independent for each variable to set it to $T$ with probability $1/32$, to set it to $F$ with probability $1/32$ and to leave it unset with probability $15/16$. Now consider the $m$ bad events, one for each clause $c$, that $c$ is either unsatisfied or has not enough unset variables (or both). Then we have 
$$\mbox{Pr}[c \mbox{ is bad}] \leq e^{-D(\frac{1}{8} || \frac{1}{16}) k} + \Big ( \frac{31}{32} \Big )^k \leq 2 \Big ( \frac{31}{32} \Big )^k $$
Here the first term follows from the Chernoff bound and represents the probability that there are not enough unset variables and the second term is the probability that the clause is unsatisfied. Moreover using the fact that $D(\frac{1}{8} || \frac{1}{16}) \geq \frac{1}{11}$ we conclude that the second term is larger than the first. Now we can once again we can appeal to the Lov\'asz Local Lemma to show the existence. Finally we can use the algorithm of Moser and Tardos \cite{MT} to find such a partial assignment in randomized polynomial time. Moreover Moser and Tardos \cite{MT} also give a deterministic algorithm that runs in time polynomial in $m$ and $n^{d^2k^2}$. 
\end{proof}

\begin{theorem}\label{thm:nonmonotonecounting}
Suppose we are given a CNF formula $\Phi$ on $n$ variables where every clause contains between $k$ and $2k$ variables. Moreover suppose that $\log d \geq 5 \log k + 20$ and $k \geq 60 \log d + 60 \log k + 300$.
Let $\mbox{OPT}$ be the number of satisfying assignments. Then there is a deterministic algorithm that outputs a quantity $\mbox{count}$ that satisfies
$$\Big ( 1 - \frac{1}{n^c} \Big) \mbox{OPT} \leq \mbox{count} \leq \Big ( 1 + \frac{1}{n^c} \Big) \mbox{OPT}$$
and runs in time polynomial in $m$ and $n^{cd^2 k}$.
\end{theorem}

\begin{proof}
First we (deterministically) find a partial assignment that meets Lemma~\ref{lemma:nonmonotonehelper}. Note that the conditions on $k$ and $d$ imply that the condition in Lemma~\ref{lemma:nonmonotonehelper} holds. Let $x_1, x_2, \dots, x_t$ be an ordering of the set variables. We define $\Phi_1, \Phi_2, \dots, \Phi_t$ in the same way as the subformula we get by substituting in the assignments for $x_1, x_2, \dots, x_{i-1}$ and simplifying to get $\Phi_i$. Again let $q_i$ be our estimate for the marginal probabilities. 

The key point is that $\Phi_{t+1}$ would be empty, because all clauses are satisfied. Moreover each clause that appears in any formula $\Phi_i$ for $1 \leq i \leq t$ has at most $2k$ variables and has at least $7k/8$ variables because it has at least that many unset variables in the partial assignment. Note that the upper and lower bound on clause sizes differ by less than a factor of $6$. Moreover we can now output 
$$\mbox{count} \triangleq 2^{n - t} \prod_{i=1}^n \Big( \frac{1}{q_i}\Big)  = \Big ( 1 \pm \frac{1}{n^c} \Big) \mbox{OPT}$$
because $\Phi_{t+1}$ has exactly $2^{n-t}$ satisfying assignments (every choice of the unset variables) and we have used the same telescoping product, but now to compute the ratio of the number of satisfying assignments to $\Phi_{t+1}$ divided by the number of satisfying assignments to $\Phi$. 
\end{proof}

\subsection{Approximate Sampling}

Here we give an algorithm to generate an assignment approximately uniformly from the set of all satisfying assignments. Again, the complication is that our oracle for approximating the marginals works only if $k$ is at least logarithmic in $d$ so we need some care in the order we choose to sample variables. First we give the algorithm:

\begin{fragment*}[t]
\caption{
\label{alg:iterrefine}{\sc Sampling Procedure}, \\ \textbf{Input:} CNF $\Phi$, oracle $F$ for approximating marginals of variables  \vspace*{0.01in}
}

\begin{enumerate} \itemsep 0pt
\small 
\item Using Lemma~\ref{lemma:marked}, label variables as marked or unmarked
\item While there is a marked variable $x$ that is unset
\item $\quad$ Sample $x$  using $F$
\item $\quad$ Initialize $V_I = \{x\}$ and $V_O$ to be all unset variables ($x$ is already set)
\item $\quad$ While there is a clause $c$ with variables in both $V_I$ and $V_O$
\item $\quad$ $\quad$ Sequentially sample its marked variables (if any) using $F$
\item $\quad$ $\quad$ Case \# 1: $c$ is satisfied
\item $\quad$ $\quad$ $\quad$ Delete $c$
\item $\quad$ $\quad$ Case \# 2: $c$ is unsatisfied
\item $\quad$ $\quad$ $\quad$ Let $S$ be all variables in $c$ (marked or unmarked)
\item $\quad$ $\quad$ $\quad$ Update $V_I \leftarrow V_I \cup S$, $V_O \leftarrow V_O \setminus S$
\item $\quad$ End
\item End
\item For each connected component of the remaining clauses
\item $\quad$ Enumerate and uniformly choose a satisfying assignment of the unset variables
\item End
\end{enumerate} 

\end{fragment*}

First, we prove that the output is close to uniform. 

\begin{lemma}\label{lemma:approxunif}
If the oracle $F$ outputs a marginal probability that is $1/n^{c+1}$ close to the true marginal distribution for each variable queried, then the output of the {\sc Sampling Procedure} is a random assignment whose distribution is $1/n^c$-close in total variation distance to the uniform distribution on all satisfying assignments. 
\end{lemma}

\begin{proof}
The proof of this lemma is in two parts. First, imagine we were instead given access to an oracle $G$ that answered each query for a marginal distribution with the exact value. Then each variable set using the oracle is chosen from the correct marginal distribution. And in the last step, the set of satisfying assignments is a cross-product of the satisfying assignments for each component. Thus the procedure would output a uniformly random assignment from the set of all satisfying assignments. 
Second, since at most $n$ variables are queried, we have that with probability at least $1-1/n^c$ all of the random decision of the procedure would be the same if we had given it answers from $G$ instead of from $F$. This now completes the proof. 
\end{proof}

The key step in the analysis of this algorithm rests on showing that with high probability each connected component is of logarithmic size. 

\begin{theorem}\label{thm:approximatesampling}
Suppose we are given a CNF formula $\Phi$ on $n$ variables where each clause contains between $k$ and $2k$ variables. Moreover suppose that $\log d \geq 5 \log k + 20$ and $k \geq 60 \log d + 60 \log k + 300$. There is an algorithm that outputs a random assignment whose distribution is $1/n^c$-close in total variation distance to the uniform distribution on all satisfying assignments. Moreover the algorithm runs in time polynomial in $m$ and $n^{cd^2 k^2}$.
\end{theorem}

\begin{proof}
The proof of this theorem uses many ideas from the coupling procedure as analyzed in Section~\ref{sec:couple}. Let $\Phi'$ be the formula at the start of some iteration of the inner WHILE loop. Then at the end of the inner WHILE loop, using Lemma~\ref{lemma:isfactored} we can write:
$$\Phi' = \Phi'_I \wedge \Phi'_O$$
where $\Phi'_I$ is a formula on the variables in $V_I$ and $\Phi'_O$ is a formula on the variables in $V_O$. In particular, no clause has variables in both because the inner WHILE loop terminated. Now we can appeal to the analysis in Theorem~\ref{thm:couplingterminates} which gives a with high probability bound on the size of $V_I$. The analysis presented in its proof is nominally for a different procedure, the {\sc Coupling Procedure}, but the inner WHILE loop of the {\sc Sampling Procedure} is identical except for the fact that there are no type 1 errors because we are building up just one assignment. Thus
$$\mbox{Pr}[|V_I| \geq 2 (6dk)^2 t] \leq \Big (\frac{1}{2} \Big )^t$$
The inner WHILE loop is run at most $n$ times and so if we choose $t \geq c \log n$ we get that with probability at least $1-1/n^{c}$ no component has size larger than $ 2c (6dk)^2 \log n$. Now the brute force search in the last step can be implemented in time polynomial in $m$ and $n^{cd^2 k^2}$, which combined with Lemma~\ref{lemma:approxunif} completes the proof.
\end{proof}

We can also now prove Corollary~\ref{corr:inference}. 

\begin{proof}
Recall, we are given a cause network and the truth assignment of each observed variable. First we do some preprocessing. If an observed variable is an OR of several hidden variables or their negation, and the observed variable is set to $F$ we know the assignment of each hidden variable on which it depends. Similarly, if an observed variable is an AND and it is set to $T$ again we know the assignment of each of its variables. For all the remaining observed variables, we know there is exactly one configuration of its variables that is prohibited so each yields a clause in a CNF formula $\Phi$. Moreover each clause depends on at least $7k/8$ variables whose truth value has not been set because the collection of observations is regular. Finally each variable is contained in at most $d$ clauses. The posterior distribution on the remaining hidden variables (whose value has not already been set) is uniform on the set of satisfying assignments to $\Phi$ and thus we can appeal to Theorem~\ref{thm:approximatesampling} to complete the proof. 
\end{proof}

\section{Conclusion}\label{sec:conclusion}

In this paper we presented a new approach for approximate counting in bounded degree systems based on bootstrapping an oracle for the marginal distribution. In fact, our approach seems to extend to non-binary approximate counting problems as well. For example, suppose we are given a set of hyperedges and our goal is to color the vertices {\tt red}, {\tt green} or {\tt blue} with the constraint that every hyperedge has at least one of each color. It is still true that the uniform distribution on satisfying colorings is locally close to the uniform distribution on all colorings in the sense of Corollary~\ref{corr:uniform}. This once again allows us to construct a coupling procedure, but now between a triple of partial colorings. The coupling can be used to give alternative ways to sample a satisfying coloring uniformly at random which in turn yields a method to certify the marginals on any vertex by solving a polynomial number of counting problems on logarithmic sized hypergraphs. We chose to work with only binary counting problems to simplify the exposition, but it remains an interesting question to understand the limits of the techniques that we introduced here.

\subsubsection*{Acknowledgements}

We are indebted to Elchanan Mossel and David Rolnick for many helpful discussions at an earlier stage of this work. We thank Heng Guo and anonymous referees for pointing out typos in some inequalities in a preliminary version, fixing which required increasing the constants in the main theorems by roughly a factor of three.


\begin{thebibliography}{99}

\bibitem{AI}
D. Achlioptas and F. Iliopoulos. 
\newblock Random walks that find perfect objects and the Lov\'asz local lemma. 
\newblock In {\em FOCS}, page 494--503, 2014. 

\bibitem{A}
N. Alon.
\newblock A parallel algorithmic version of the local lemma.
\newblock In {\em Random Structures and Algorithms}, 2(4):367--378, 1991. 

\bibitem{B}
J. Beck.
\newblock An algorithmic approach to the Lov\'asz local lemma.
\newblock {\em Random Structure and Algorithms}, 2(4):343--365, 1991. 

\bibitem{BGGGS}
I. Bez\'akova, A. Galanis, L. A. Goldberg, H. Guo and D. \u{S}tefankovi\u{c}.
\newblock Approximation via correlation decay when strong spatial mixing fails. 
\newblock In {\em ArXiv}:1510.09193, 2016. 

\bibitem{BDK1}
M. Bordewich, M. Dyer and M. Karpinski. 
\newblock Stopping times, metrics and approximate counting.
\newblock In {\em ICALP}, pages 108--119, 2006. 

\bibitem{BMS}
G. Bresler, E. Mossel and A. Sly.
\newblock Reconstruction of Markov random fields from samples: Some observations and algorithms. 
\newblock {\em SIAM Journal on Computing}, 42(2):563--578, 2013.

\bibitem{DL1}
P. Dagum and M. Luby.
\newblock Approximating probabilistic inference in Bayesian belief networks is $NP$-hard.
\newblock {\em Artificial Intelligence}, 60:141--153, 1993. 

\bibitem{DL2}
P. Dagum and M. Luby.
\newblock An optimal approximation algorithm for Bayesian inference. 
\newblock {\em Artificial Intelligence}, 93(1-2):1--27, 1997. 

\bibitem{EL}
P. Erd\"os and L. Lov\'asz.
\newblock Problems and results on $3$-chromatic hypergraphs and some related questions. 
\newblock In {\em Infinite and Finite Sets}, pages 609--627, 1975. 

\bibitem{GSV}
A. Galanis, D. \u{S}tefankovi\u{c} and E. Vigoda.
\newblock Inapproximability of the partition function for the antiferromagnetic Ising and hard-core models. 
\newblock {\em Combinatorics, Probability and Computing}, 25(4):500--559, 2016. 

\bibitem{GK}
D. Gamarnik and D. Katz.
\newblock Correlation decay and deterministic FPTAS for counting colorings of a graph.
\newblock {\em Journal of Discrete Algorithms}, 12:29--47, 2012. 

\bibitem{GJL}
H. Guo, M. Jerrum and J. Liu.
\newblock Uniform Sampling through the Lov\'asz Local Lemma.
\newblock In {\em STOC} 2017, to appear. 

\bibitem{HSS}
B. Haeupler, B. Saha and A. Srinivasan.
\newblock New constructive aspects of the Lov\'asz Local Lemma. 
\newblock In {\em Journal of the ACM}, 58(6): 28, 2011. 

\bibitem{HS}
D. Harris and A. Srinivasan.
\newblock A constructive algorithm for the Lov\'asz local lemma on permutations.
\newblock In {\em SODA}, pages 907--925, 2014. 

\bibitem{HSV}
N. Harvey, P. Srivastava and J. Vondr\'ak.
\newblock Computing the independence polynomial in Shearer's region for the LLL. 
\newblock {\em ArXiv}:1608.02282, 2016. 

\bibitem{HV}
N. Harvey and J. Vondr\'ak. 
\newblock An algorithmic proof of the Lov\'asz local lemma via resampling oracles.
\newblock In {\em FOCS}, pages 1327--1346, 2015. 

\bibitem{HSZ}
J. Hermon, A. Sly and Y. Zhang.
\newblock Rapid mixing of hypergraph independent set. 
\newblock In {\em ArXiv:}1610.07999, 2016. 

\bibitem{JVV}
M. Jerrum, L. Valiant and V. Vazirani.
\newblock Random generation of combinatorial structures from a uniform distribution.
\newblock {\em Theoretical Computer Science}, 43:169--188, 1986. 

\bibitem{K}
D. Knuth.
\newblock {\em The Art of Computer Programming}, Vol I.
\newblock Addison Wesley, London, page 396 (exercise 11), 1969.

\bibitem{Ko}
V. Kolmogorov. 
\newblock Commutativity in the algorithmic Lov\'asz local lemma. 
\newblock In {\em FOCS} 2016, to appear. 

\bibitem{LL}
J. Liu and P. Lu.
\newblock FPTAS for counting monotone CNF.
\newblock In {\em SODA}, pages 1531--1548, 2015. 

\bibitem{MT}
R. Moser and G. Tardos.
\newblock A constructive proof of the Lov\'asz Local Lemma.
\newblock In {\em Journal of the ACM}, 57(2):1--15, 2010. 

\bibitem{NT}
C. Nair and P. Tetali.
\newblock The correlation decay (CD) tree and strong spatial mixing in multi-spin systems.
\newblock {\em ArXiv:}0701494, 2007. 

\bibitem{JS}
A. Sinclair and M. Jerrum.
\newblock Approximately counting, uniform generation and rapidly mixing markov chains.
\newblock {\em Information and Computation}, 82:93--133, 1989. 

\bibitem{SST}
A. Sinclair, P. Srivastava and M. Thurley.
\newblock Approximation algorithms for two-state anti-ferromagnetic spin systems. 
\newblock {\em Journal of Statistical Physics}, 155(4):666--686, 2014. 

\bibitem{S}
A. Sly.
\newblock Computational transition at the uniqueness threshold.
\newblock In {\em FOCS}, pages 287--296, 2010.

\bibitem{SS}
A. Sly and N. Sun.
\newblock Counting in two-spin models on $d$-regular graphs. 
\newblock {\em Annals of Probability}, 42(6):2383--2416, 2014. 

\bibitem{W}
D. Weitz.
\newblock Counting independent sets up to the tree threshold. 
\newblock In {\em STOC}, pages 140--149, 2006.

\end{thebibliography}
\end{document}